\documentclass[12pt]{article}

\usepackage[T1]{fontenc}
\usepackage[utf8]{inputenc}
\usepackage{geometry}
\usepackage{setspace}
\usepackage{amsmath, amssymb, amsthm, bm}

\usepackage{graphicx}
\usepackage{xcolor}

\usepackage{natbib}

\usepackage{enumitem}
\usepackage{needspace}

\usepackage[titletoc,title]{appendix}

\usepackage{float}
\usepackage{hyperref}
\hypersetup{
    colorlinks=true,
    citecolor=blue,
    linkcolor=blue,
    urlcolor=blue
}

\theoremstyle{plain}
\newtheorem{theorem}{Theorem}[section]
\newtheorem{lemma}[theorem]{Lemma}
\newtheorem{proposition}[theorem]{Proposition}

\newtheorem*{lemma*}{Lemma}
\newtheorem*{proposition*}{Proposition}
\newtheorem*{corollary*}{Corollary}

\theoremstyle{definition}
\newtheorem{definition}[theorem]{Definition}

\theoremstyle{remark}

\usepackage[normalem]{ulem}
\usepackage{cancel}
\usepackage{booktabs}

\usepackage{comment}

\newcommand{\vs}{\mathbf{s}}           
\newcommand{\vd}{\boldsymbol{\delta}}  
\newcommand{\vx}{\mathbf{x}}          
\newcommand{\vz}{\mathbf{z}}          
\newcommand{\vG}{\mathbf{G}}          
\newcommand{\vH}{\mathbf{H}}          
\newcommand{\vone}{\mathbf{1}}        
\newcommand{\vbeta}{\boldsymbol{\beta}}      
\newcommand{\vsigma}{\boldsymbol{\sigma}}    

\newcommand{\cJ}{\mathcal{J}}         
\newcommand{\cG}{\mathcal{G}}         
\newcommand{\Rpp}{\mathbb{R}_{++}}    
\newcommand{\R}{\mathbb{R}}    

\newcommand{\JlnG}{\mathbf{J}_{\ln\vG}^{\vs}}  

\title{Demand Models for Market-Level Data with Closed-Form Inverses\thanks{This paper builds on the theoretical part of an earlier version of a paper later published as Fosgerau, Monardo, and de Palma (2024), The Inverse Product Differentiation Logit Model, \emph{American Economic Journal: Microeconomics, 16}(4), 329--370. It also supersedes our earlier unpublished working papers: Mogens Fosgerau and Andr\'e de Palma (2015), \emph{Demand Systems for Market Shares}, MPRA Paper 62106; Mogens Fosgerau and Andr\'e de Palma (2016), \emph{Generalized Entropy Models}, HAL Working Paper hal-01291347; and Andr\'e de Palma and Julien Monardo (2017), \emph{The General Nesting Logit (GNL) Model using Aggregate Data}, HAL Working Paper hal-01552455. This paper has benefited greatly from discussions with Aureo de Paula, Phil Haile, Ali Horta\c{c}su, and Jonas Lieber. The third author acknowledges financial support from the French National Research Agency (ANR) under the France 2030 programme through project HARMONIC (grant ANR-24-PEMO-0003).}}
\author{Julien Monardo\thanks{University of Bristol. Email: \texttt{julien.monardo@bristol.ac.uk}.}
\and
Mogens Fosgerau\thanks{University of Copenhagen. Email: \texttt{mogens.fosgerau@econ.ku.dk}.}
\and
André de Palma\thanks{CY Cergy Paris Université. Email: \texttt{andre.de-palma@cyu.fr}.}
}
\date{November 24, 2026}

\begin{document}
\maketitle

\begin{abstract}
We develop a tractable approach to specifying demand models for market-level data through closed-form inverse market share functions. The class includes specifications that can be estimated by linear instrumental variables regression while accommodating rich substitution patterns. They are built through a generator analogous to McFadden's generalized extreme value generating function, but acting on market shares. Constructive results allow arbitrary nesting structures, including overlapping nests and partial membership, yielding inverse-share analogs of generalized extreme value models. We establish equivalence between our generator representation and quasi-linear utility maximization under regularity conditions. Under mild conditions, additive random utility models form a strict subclass.\\[1ex]
\end{abstract}

\newpage
\section{Introduction}

Estimating demand for differentiated products with market-level data generally entails a trade-off between tractability and flexibility. Models with linear-in-parameters closed-form inverses, such as logit and nested logit, allow linear instrumental variables estimation but restrict substitution patterns. More flexible models, such as random coefficient logit, accommodate richer substitution but lack closed-form inverses and therefore require repeated numerical inversion for estimation.

We develop an approach that relaxes this trade-off by specifying demand through closed-form inverse market share functions. These models accommodate rich substitution patterns, remain consistent with utility maximization, and can be estimated by instrumental variables regression. Some nesting specifications allow linear estimation, avoiding numerical inversion, starting values, and local optima. These features make them attractive for counterfactual and welfare analysis with many products or markets.

We specify inverse market share functions via a generator; it is analogous to the generalized extreme value (GEV) generating function of \citet{mcfadden1978modeling}, but acts on market shares rather than utilities. The resulting models retain the key ingredients of the demand estimation framework of \cite{berry1994estimating} and \cite{BLP1995}: products are differentiated by observed and unobserved characteristics, and instrumental variables address endogeneity arising from the latter. The inverse market share function and its derivatives allow us to compute price elasticities and diversion ratios and conduct counterfactual and welfare analysis.

We provide constructive results accommodating arbitrary nesting structures, including overlapping nests and partial membership. These results yield inverse share analogs of well-known GEV models \citep[see][Chapter 4.6]{train2009discrete}. With fixed memberships, these models preserve the computational simplicity of logit and nested logit while allowing richer substitution patterns through more flexible nesting structures. With nonnegative nesting parameters, hierarchical full membership preserves substitutability, whereas overlap can generate complementarity. A second-order expansion around logit explains how nesting can generate complementarity in demand. Pairwise nesting can match any admissible local substitution pattern among inside products.

Nests provide a parsimonious representation of substitution by summarizing many product-pair substitution relationships through a small number of groups. They can represent unobserved dimensions of local substitution: membership weights measure products' association with each dimension, and nesting parameters determine its contribution. Nests can also proxy for heterogeneity associated with continuous characteristics and unobserved preferences \citep{grigolon2014nested}. Choosing nests restricts substitution, just as choosing which characteristics carry random coefficients does in random coefficient logit. With fixed memberships, linear estimation makes it inexpensive to compare alternative structures through cross-validation or non-nested tests.\footnote{See, e.g., \citet{komiyama2018cross} and \citet{rivers2002model}.}\par

We establish an equivalence between demand models defined by our generators and a class of quasi-linear utility maximization problems \citep{mcfadden2012theory}. Every model admits a representative consumer interpretation, as do the logit, nested logit, and inverse product differentiation logit (IPDL) models \citep{anderson1988representative,verboven1996,fosgerau2024}. The representative consumer problem has the form that \citet{allen2019} derive by aggregating heterogeneous utility-maximizing consumers under their assumptions. Lastly, every additive random utility model (ARUM) satisfying mild regularity assumptions admits our generator representation, whereas the converse does not hold.

This paper relates to three strands of literature. First, it builds on demand estimation for differentiated products using market-level data, following \cite{berry1994estimating} and \cite{BLP1995}.  Unlike random coefficient logit models and nesting-based GEV models beyond nested logit, our nesting specifications allow linear estimation with fixed memberships.

Second, it relates to empirical work balancing tractability and flexibility. One approach reduces dimensionality by summarizing substitution patterns with lower-dimensional objects \citep{pinkse2004mergers,fosgerau2024,fally2026,fallyligon2026}. Another approach reduces dimensionality by exploiting permutation invariance \citep{singh2023,Monardo2026,Nevo2026}. Other papers simplify random coefficient logit estimation \citep{lu2023,salanie2022}.  We establish necessary conditions for utility consistency, a necessary and sufficient condition under hierarchical full membership, and results on substitution, complementarity, and the local flexibility of pairwise nesting. {Our nesting constructions build on the IPDL model of \citet{fosgerau2024} and the related partial membership specification of \citet{fosgerauNielsen2021similarity}.} 

Third, it relates to theory linking continuous and discrete choice. \citet{anderson1988representative} and \citet{verboven1996} derive representative consumer direct utilities for logit and nested logit demand, respectively. \citet{dube2025discrete} derive analogous utilities for discrete/continuous choice models. \citet{monardo2025functional} characterizes a large family of utility-consistent discrete choice demands beyond ARUMs. \citet{hofbauer2002global} show that ARUM demands admit a representative consumer representation. \citet{allen2019} link representative consumer models to heterogeneous, utility-maximizing consumers. \citet{fosgerau2016discrete} connect ARUMs to generators in a rational inattention framework.  {We establish equivalence under weaker boundary conditions, extending the utility-to-generator construction in \citet{fosgerauKristensenNielsen2021} and proving the converse.}

The paper is structured as follows. Section \ref{Section:Class} introduces the framework and the inverse market share functions. Section \ref{Sec:IdentandEstim} discusses identification and estimation. Section \ref{Section:Utility} studies microfoundations. Section \ref{sec:ccl} concludes. 

\section{Demand with Closed-Form Inverses} \label{Section:Class}

\subsection{Framework}

Our framework follows the empirical literature \citep{BH2016,BH2021}. In each market $t=1,\ldots,T$, consumers choose among differentiated products in the set 
$\cJ=\{0,1,\ldots,J\}$, where $j=0$ is the outside option. Each product $j$ is defined by its characteristics $\vx_{jt}$, price $p_{jt}$, and an
unobserved characteristic term $\xi_{jt}$. The $\xi_{jt}$ terms are the model's errors, summarizing product/market characteristics unobserved by the modeler but potentially observed by consumers and firms. Each product $j$ is associated with a linear index 
\begin{equation}\label{eq:Delta}
\delta_{jt}
= \vx_{jt}^{\top}\vbeta-\alpha p_{jt}+\xi_{jt},
\qquad
\delta_{0t}=0,
\end{equation}
where $\vbeta\in\R^K$ and $\alpha >0$. Let $\vd_t=(\delta_{0t}, \ldots,\delta_{Jt})^\top$ and define $\mathbb R^{J+1}_0
=
\{\vd\in\mathbb R^{J+1}:\delta_0=0\}$.

Let $\Delta_J=\{\vs\in[0,\infty)^{J+1}:\sum_{j}s_j=1\}$ and $\Delta_J^\circ=\{\vs\in(0,\infty)^{J+1}:\sum_{j}s_j=1\}$ denote the sets of market share vectors and positive market share vectors, respectively. Aggregating choices in market $t$ gives the share vector $\vs_t=(s_{0t},\ldots,s_{Jt})^\top \in \Delta_J^\circ$. Assume choices are characterized by the market share function $\boldsymbol{\sigma}:\R^{J+1}\rightarrow\Delta_J^\circ$ with $\boldsymbol{\sigma}(\vd_t)=(\sigma_0(\vd_t),\ldots,\sigma_J(\vd_t))^\top$.

The market share function maps product indexes to observed shares:
\begin{equation}
    s_{jt} = \sigma_j (\vd_t),\quad j \in \cJ.
\end{equation}

When it exists, the inverse market share function 
$\boldsymbol{\sigma}^{-1}:\Delta_J^\circ\rightarrow \R_0^{J+1}$, with $\boldsymbol{\sigma}^{-1}(\vs_t)=(\sigma_0^{-1}(\vs_t),\ldots,\sigma_J^{-1}(\vs_t))^\top$, satisfies $\delta_{jt}=\sigma_j^{-1}(\vs_t)$. Substituting \eqref{eq:Delta} yields
\begin{equation}\label{eq:xii}
    \xi_{jt}
    =\sigma_j^{-1}(\vs_t)
    -\vx_{jt}^{\top}\vbeta
    +\alpha p_{jt},\quad j=1,\ldots,J.
\end{equation}

Prices and shares are allowed to be endogenous, i.e., correlated with the error terms $\xi_{jt}$. Following \cite{berry1994estimating}, estimation uses the conditional moments
\begin{equation}
    \mathbb{E}[\xi_{jt}|\vx_t,\vz_t] = 0, \quad j = 1,\ldots,J,
\end{equation}
where $\vx_t=(\vx_{1t},\ldots,\vx_{Jt})$ collects observed inside-product characteristics and $\vz_t=(\vz_{1t},\ldots,\vz_{Jt})$ collects excluded instruments. 


\subsection{The Inverse Market Share Functions \label{Sec:ClassModels}}

We suppress the market subscript $t$. For $\vG:\R^{J+1}\to\R^{J+1}$, let $\mathbf J_{\vG}^{\vs}(\tilde{\vs})=[\partial G_i(\tilde{\vs})/\partial s_j]_{i,j=0}^{J}$ denote its Jacobian. Proofs for Section~\ref{Section:Class} are given in Appendix~\ref{Appendix:ProofSec2}.

We build our inverse market share functions from a generator.

\begin{definition}\label{Def:Generator}
A generator is a continuously differentiable function $\vG:\Rpp^{J+1}\rightarrow\Rpp^{J+1}$, with $\vG(\vs)=(G_0(\vs),\ldots,G_J(\vs))^\top$, such that:
(i) $\vG$ is homogeneous of degree $1$;
(ii) $\|\ln\vG(\vs^n)\|\to\infty$ for every sequence $\vs^n\in\Delta_J^\circ$ converging to a point of $\operatorname{bd}(\Delta_J)$;
(iii) the Jacobian $\mathbf{J}_{\ln\vG}^{\vs}(\vs)$ is symmetric and
positive definite on $\Delta_J^\circ$.
\end{definition}

Generators are invertible on $\mathbb{R}_{++}^{J+1}$ (Lemma~\ref{prop:inversion}, Appendix~\ref{Appendix:Technical}). Let $\vone_{J+1}=(1,\ldots,1)^\top\in\R^{J+1}$ denote the vector of ones. Homogeneity gives $\ln\vG(\lambda\vs)=\ln\vG(\vs)+(\ln\lambda)\vone_{J+1}$ and $\JlnG(\lambda\vs)=\lambda^{-1}\JlnG(\vs)$ for $\lambda>0$. Thus symmetry and positive definiteness extend from $\Delta_J^\circ$ to $\Rpp^{J+1}$.  
Given a generator, we define
\begin{equation}\label{eq:inversemarket}
\sigma_j^{-1}(\vs)
=
\ln G_j(\vs) + c,\quad j \in\cJ
\end{equation}
where $c$ is a market-specific constant determined by the normalization
$\delta_0=0$.

Let $\vH = \vG^{-1}:\Rpp^{J+1}\rightarrow \Rpp^{J+1}$, with $\mathbf H (e^{\vd}) =(H_0(e^{\vd}),\ldots,H_J(e^{\vd}))^\top =  \mathbf G ^{-1}(e^{\vd})$, denote the inverse of the generator $\vG$. Inverting \eqref{eq:inversemarket} on $\Delta_J^\circ$ gives the market share function
\begin{equation}\label{eq:demandsystem}
\sigma_j(\vd)
=
\frac{H_j(e^{\vd})}
{\sum_{k\in\cJ}H_k(e^{\vd})}.
\end{equation}
It is invariant to adding a constant to all product indexes: $\boldsymbol{\sigma}(\vd+\tau\vone_{J+1}) = \boldsymbol{\sigma}(\vd)$ for any scalar $\tau\in\R$.

Let $\vs=\boldsymbol{\sigma}(\vd)$. The substitution matrix is
\begin{equation}\label{eq:substitution-matrix}
\mathbf{J}_{\boldsymbol{\sigma}}^{\vd}{(\vd)
=
\bigl[\mathbf{J}_{\ln\vG}^{\vs}(\vs)\bigr]^{-1}
-
\vs\vs^\top,}
\end{equation}
which is symmetric and positive semidefinite.

The function $CS$ defined by 
\begin{equation}\label{eq:CSfunc}
CS(\vd)=\ln\!\left(\sum_{k\in\cJ}H_k(e^{\vd})\right)
\end{equation}
satisfies $\partial CS(\vd)/\partial\delta_j=\sigma_j(\vd)$. It is convex and
increasing, and it satisfies $CS(\vd+\tau\vone_{J+1})=CS(\vd)+\tau$ for any scalar
$\tau\in\R$. Section~\ref{Section:Utility} shows that it is, up to the income term,
the indirect utility associated with a quasi-linear direct utility
function, giving it a consumer surplus interpretation in index units.

The market share function $\boldsymbol{\sigma}$ in \eqref{eq:demandsystem} is generally not available in closed form.  
The substitution matrix gives elasticities and diversion ratios. Its symmetry and positive semidefiniteness imply the weak law of demand, $\partial \sigma_j(\vd)/\partial\delta_j\geq 0$, and the existence of a convex surplus
function $CS$ satisfying
$\partial CS(\vd)/\partial\delta_j=\sigma_j(\vd)$ \citep[see, e.g.,][Section 3.A.9]{sandholm2010population}.  Substitution depends on characteristics only through $\vd$, as in GEV models, whereas random coefficient logit allows direct dependence on continuous characteristics. Nests can nevertheless proxy for these characteristics.

\subsection{Constructing Generators through Nesting}

\subsubsection{Nesting-Based Generators}\label{sec:nesting-based-generators}

Let $\cG\subseteq2^{\cJ}$ be a
collection of nests, each containing at least two products, with parameters $\mu_g\in\R$
for $g\in\cG$, and define $\vG$ by
\begin{equation}\label{eq:generatorGN}
\ln G_j(\vs)
=
\Bigl(1-\!\!\sum_{g\in\cG:j\in g}\mu_g\Bigr)\ln s_j
+
\sum_{g\in\cG:j\in g}\mu_g
\ln\Bigl(\sum_{i\in g}s_i\Bigr).
\end{equation}
This construction covers logit, nested logit models with any number of levels, and the IPDL model.

Nest membership may also be partial. Let $\omega_{jg}\ge0$ measure the degree to which
product $j$ belongs to nest $g$, with at least two positive weights in each nest, and
define $\vG$ by
\begin{equation}\label{eq:ConstructionMember}
\ln G_j(\vs)
=
\left(1-\sum_{g \in\cG} \mu_g \omega_{jg}\right)\ln s_j
+
\sum_{g \in\cG} \mu_g \omega_{jg}\,
\ln\!\left(\sum_{i\in\cJ}\omega_{ig}s_i\right).
\end{equation}
This reduces to \eqref{eq:generatorGN} when $\omega_{jg}=1\{j\in g\}$.

Assume that the outside good belongs to no nest, $\omega_{0g}=0$ for every $g\in\cG$. Both constructions define a generator if
$\sum_{g\in\cG:\,\mu_g>0}\mu_g\omega_{jg}<1$ for every $j\in\cJ$. A necessary condition is
$b_j=1-\sum_{g\in\cG}\mu_g\omega_{jg}>0$ for every $j\in\cJ$. The two conditions coincide
for nonnegative parameters. These conditions are established in Proposition \ref{prop:direct-nesting}. \citet{fosgerauNielsen2021similarity} introduce an equivalent specification (up to a linear term and scaling) with signed parameters and prove the sufficient condition. {\color{black}\hypersetup{allcolors=black}
For full membership with hierarchical nests, meaning that any two nests are disjoint or one contains the other, a necessary and sufficient condition for construction {\hypersetup{linkcolor=blue}\eqref{eq:generatorGN}} to define a generator is $\sum_{h\in\cG:\,g\subseteq h}\mu_h<1$ for every $g\in\cG$.
}

Membership weights may be estimated or fixed using observed product characteristics or an externally specified overlap structure \citep{small1987discrete,bresnahan1997market,grigolon2021blurred}.

As shown in Appendix~\ref{Appendix:ProofSec2}, because the outside good belongs to no nest and $\delta_0=0$, consumer surplus is $CS(\vd)=-\ln s_0$.
Nesting therefore affects welfare through counterfactual changes
in total inside demand, $1-s_0$. 

\subsubsection{Substitution Patterns}\label{sec:substitution-patterns}

Differentiating \eqref{eq:ConstructionMember} with respect to $\vs$ gives
\[
\JlnG(\vs)
=\operatorname{diag}\!\left(\frac{b_j}{s_j}\right)
+\sum_g\mu_g\frac{\boldsymbol\omega_g\boldsymbol\omega_g^\top}{S_g},
\]
where $b_j=1-\sum_g\mu_g\omega_{jg}$, $S_g=\sum_i\omega_{ig}s_i$, and $\boldsymbol\omega_g=(\omega_{0g},\ldots,\omega_{Jg})^\top$.
Using \eqref{eq:substitution-matrix}, products $j\ne k$ are substitutes in demand when $\bigl[\JlnG(\vs)^{-1}\bigr]_{jk}\leq s_js_k$ and complements in demand otherwise. Nonpositive off-diagonal entries of $[\JlnG(\vs)]^{-1}$ imply that products are at least as substitutable as under logit at the same shares.

The following proposition relates substitution to nesting.

\begin{proposition}\label{prop:substitution-structure}
Consider a finite collection of nests with nonnegative membership weights and at least two positive weights in each nest. A product belongs to a nest if its weight is positive.
\begin{enumerate}
\item[(i)]  {\color{black}\hypersetup{allcolors=black}For full membership {\hypersetup{linkcolor=blue}\eqref{eq:generatorGN}} with hierarchical nests, if $\mu_g\geq0$, $\sum_{g:\,j\in g}\mu_g<1$ for every $j\in\cJ$, then $\frac{\partial\sigma_j(\vd)}{\partial\delta_k}\leq-s_js_k<0$ for every $j\ne k$ and every $\vs\in\Delta_J^\circ$.}
\item[(ii)] For partial membership \eqref{eq:ConstructionMember}, if two nests overlap without either containing the other, then some positive nesting parameters satisfying $\sum_g\mu_g\omega_{jg}<1$ for every $j\in\cJ$ and some $\vs\in\Delta_J^\circ$ make two products complements in demand.
\end{enumerate}
\end{proposition}

Part (i) gives substitutes in demand for hierarchical full membership. The assumptions in part (i) have a nested logit interpretation:  {\color{black}\hypersetup{allcolors=black}in the three-level model in Table~\ref{tab:GEV}, our nesting parameters satisfy $\mu_1=\lambda_1-\lambda_2$ and $\mu_2=1-\lambda_1$. Their nonnegativity and $\mu_1+\mu_2<1$ give the usual nested logit restrictions $0<\lambda_2\leq\lambda_1\leq1$.} Negative nesting parameters can reverse this ordering. Part (ii) shows that overlap allows, but does not imply, complementarity: with overlap, parameter signs alone do not determine cross-price effects. 

The next proposition explains how nesting can generate complementarity in demand.

\Needspace{12\baselineskip}
\begin{proposition}\label{prop:substitution-expansion}
{Consider construction \eqref{eq:ConstructionMember} with} a {fixed }finite collection of nests, nonnegative membership weights with at least two positive weights {per} nest, and {fixed shares }$\vs\in\Delta_J^\circ$. Let
\[
w_{jk}=\sum_g\mu_g\frac{\omega_{jg}\omega_{kg}s_js_k}{S_g},\qquad
w_j=\sum_{i\ne j}w_{ji},\qquad j\ne k.
\]

With $\vsigma(\vd)=\vs$, the second-order Taylor polynomial approximation to $\partial\sigma_j(\vd)/\partial\delta_k$ around $\boldsymbol\mu=\mathbf0$, holding shares and membership weights fixed, is
\[
\underbrace{-s_js_k}_{\text{logit}}
+\underbrace{(-w_{jk})}_{\text{first order}}
+\underbrace{\left[\sum_{i\ne j,k}\frac{w_{ji}w_{ik}}{s_i}
-w_{jk}\left(\frac{w_j}{s_j}+\frac{w_k}{s_k}\right)\right]}_{\text{second order}}
,\qquad j\ne k.
\]

\end{proposition}

The weight \(w_{jk}\) measures the strength and sign of the link between products \(j\) and \(k\) through shared nests at \(\vs\), while \(w_j\) measures product \(j\)’s overall connection to the other products. When all nesting parameters are nonnegative, sharing a nest strengthens substitution through the first-order term $-w_{jk}$. Substitution through a third product $i$ enters through the second-order term $w_{ji}w_{ik}/s_i$ and can offset this effect, contributing to complementarity in demand.  At fixed interior shares, sufficiently small parameters preserve substitutability. 

Similarly, \citet{ogaki1990} decomposes the total price response into a direct effect measuring substitution between two products with other shares held fixed, and an indirect effect capturing how adjustments in those shares modify this response. Proposition~\ref{prop:substitution-expansion} describes how nesting changes the total price response relative to logit. For inside products $j\ne k$, let $\partial\sigma_j/\partial p_k=S^d_{jk}+S^i_{jk}$, where $S^d_{jk}$ and $S^i_{jk}$ are the direct and indirect effects, respectively. Then,
\[
\begin{aligned}
S^d_{jk}
&=\frac{\alpha}{b_j/s_j+b_k/s_k+\sum_g\mu_g(\omega_{jg}-\omega_{kg})^2/S_g}>0,\\
S^i_{jk}&=-\alpha\frac{\partial\sigma_j(\vd)}{\partial\delta_k}-S^d_{jk}.
\end{aligned}
\]
Since $S^d_{jk}>0$, complementarity arises only through the indirect effect, when $S^i_{jk}<-S^d_{jk}$.

The next proposition shows that pairwise nesting is locally flexible.

\begin{proposition}\label{prop:combined-flexibility}
Let $J\geq2$, fix $\vs\in\Delta_J^\circ$, and let $\mathbf A^*\in\R^{(J+1)\times(J+1)}$ be symmetric positive semidefinite of rank $J$, with $\mathbf A^*\vone_{J+1}=\mathbf0$ and $A^*_{0j}=-s_0s_j$ for $j=1,\ldots,J$. Consider \eqref{eq:ConstructionMember} with one nest for each pair of inside products, unit memberships, and parameters $\mu_{jk}=\mu_{kj}$. Its substitution matrix at $\vs$, $[\JlnG(\vs)]^{-1}-\vs\vs^\top$, equals $\mathbf A^*$ if and only if
\[
\mu_{jk}=(s_j+s_k)\bigl[(\mathbf A^*+\vs\vs^\top)^{-1}\bigr]_{jk},\qquad 1\leq j<k\leq J.
\]
These parameters define a generator if $\sum_{k\neq j:\,\mu_{jk}>0}\mu_{jk}<1$ for $j=1,\ldots,J$, and only if $\sum_{k\neq j}\mu_{jk}<1$ for $j=1,\ldots,J$.
\end{proposition}

 {\color{black}\hypersetup{allcolors=black}The proposition shows that pairwise nesting can match any interior share vector and admissible inside-product substitution matrix. It is therefore flexible in the sense of \citet{diewert1974}, except for substitution with the outside good: excluding it from all nests fixes $\frac{\partial\sigma_0(\vd)}{\partial\delta_j}=-s_0s_j$ at its logit value.} Logit corresponds to $\mu_{jk}=0$, so substitution patterns close to logit yield small parameters that satisfy the sufficient condition $\sum_{k\neq j:\,\mu_{jk}>0}\mu_{jk}<1$, and the matched model is then a generator. For target substitution patterns away from logit, the matched parameters need not define
a generator. The model then reproduces $\mathbf A^*$ at $\vs$ but is not a valid demand system at all share vectors. 

\subsubsection{GEV Counterparts}\label{sec:gev-counterparts}

GEV models are ARUMs with multivariate extreme value shocks \citep{mcfadden1978modeling, anderson1992discrete, fosgerau2013choice}. They allow varied nesting structures \citep[see, e.g.,][]{small1987discrete, chu1989paired,bresnahan1997market,vovsha1997application,koppelman2000paired,wen2001generalized,Ben-Akiva1999a,davis2014flexible,grigolon2021blurred}.

Table~\ref{tab:GEV} pairs GEV models with our nesting constructions. GEV models embed nesting in market shares, whereas we embed it in the inverse market share function. In particular, the IPDL model of \citet{fosgerau2024}, like the product differentiation logit model of \citet{bresnahan1997market}, relaxes the hierarchical nesting restriction of the nested logit models but does not allow partial membership. The flexible inverse logit model studied by \citet{Monardo2026} uses the pairwise nesting structure of the paired combinatorial logit \citep{chu1989paired,koppelman2000paired}. Our partial membership construction parallels the cross-nested logit model \citep{bierlaire2006theoretical}.

\clearpage
\begin{table}[H]
\setstretch{0.95}
\centering
\caption{GEV generator vs our generator}\label{tab:GEV}
\setlength\tabcolsep{6pt}
\resizebox{\textwidth}{!}{\begin{tabular}{@{}c | c@{}}
\cmidrule(lr){1-2}
\textbf{GEV generator} & \textbf{Our generator} $\ln G_j(\vs)$\\[0.5ex]
\cmidrule(lr){1-2}

\emph{Logit} &
\emph{Logit} \\
$\displaystyle \sum_{j=0}^J e^{\delta_j}$ &
$\displaystyle\;\ln(s_j)$ \\[0.5ex]

\cmidrule(lr){1-2}

\emph{Two-Level Nested Logit} &
\emph{Two-Level Nested Logit} \\
$\displaystyle
\begin{gathered}
\sum_{g=1}^{N}\left(\sum_{j\in g}e^{\delta_j/\lambda}\right)^{\lambda}
\end{gathered}$ &
$\displaystyle\;(1-\mu)\ln(s_j)+\mu\ln\!\Bigl(\sum_{k\in g(j)}s_k\Bigr)$ \\
$\begin{gathered}
0<\lambda\leq1
\end{gathered}$ &
$\mu=1-\lambda,\quad0\leq\mu<1$ \\[0.5ex]

\cmidrule(lr){1-2}

\emph{Three-Level Nested Logit} &
\emph{Three-Level Nested Logit} \\
$\displaystyle \sum_{g=1}^{N}
\left(
\sum_{h=1}^{N_*}
\left(
\left(
\sum_{j\in\cG_{hg}}e^{\delta_j/\lambda_2}
\right)^{\lambda_2}
\right)^{1/\lambda_1}
\right)^{\lambda_1}$ &
$\displaystyle
(1-\mu_1-\mu_2)\ln s_j
+\mu_1\ln\!\Biggl(\sum_{k\in g_1(j)}s_k\Biggr)
+\mu_2\ln\!\Biggl(\sum_{k\in g_2(j)}s_k\Biggr)$ \\[0.5ex]
$\begin{gathered}
0<\lambda_2\leq\lambda_1\leq1
\end{gathered}$ &
$\begin{gathered}
\mu_1=\lambda_1-\lambda_2\geq0,\quad
\mu_2=1-\lambda_1\geq0,\quad
\mu_1+\mu_2<1
\end{gathered}$ \\[0.5ex]

\cmidrule(lr){1-2}

\emph{Product Differentiation Logit} &
\emph{IPDL} \\
$\displaystyle \sum_{g=1}^{N}a_g
\left(
\sum_{h=1}^{N_g}
\left(
\sum_{j\in\cG_{hg}}e^{\delta_j/\lambda_g}
\right)^{\lambda_g}
\right)$ &
$\displaystyle
\begin{gathered}
\left(1-\sum_{g=1}^N\mu_g\right)\ln s_j
+\sum_{g=1}^N\mu_g
\ln\!\Biggl(\sum_{k\in\cG_{h_g(j),\,g}}s_k\Biggr)
\end{gathered}$ \\[0.5ex]
$a_g>0$, $\sum_{g=1}^{N}a_g=1$, $0<\lambda_g\leq1$ $\forall g$ &
$\mu_g\geq0$ for all $g$; $\sum_{g=1}^{N}\mu_g<1$ \\[0.5ex]

\cmidrule(lr){1-2}

\emph{Paired Combinatorial Logit} &
\emph{The Flexible Inverse Logit} \\
$\displaystyle \sum_{j=0}^{J-1}\sum_{k=j+1}^{J}
\left(e^{\delta_j/\lambda_{jk}}+e^{\delta_k/\lambda_{jk}}\right)^{\lambda_{jk}}$ &
$\displaystyle\;
\Biggl(1-\sum_{k\neq j}\mu_{jk}\Biggr)\ln s_j
+\sum_{k\neq j}\mu_{jk}\ln(s_j+s_k)$ \\[0.5ex]
$0<\lambda_{jk}\leq1$ $\forall j\neq k$ &
$\mu_{jk}=\mu_{kj}
\geq0\ \forall j\neq k,\quad
\sum_{k\neq j}\mu_{jk}<1\ \forall j$ \\[0.5ex]

\cmidrule(lr){1-2}

\emph{Cross-Nested Logit} &
\emph{Partial Membership} \\
$\displaystyle \sum_{m\in\mathcal M}
\left(
\sum_{j\in\cJ}\omega_{jm}e^{\delta_j/\lambda_m}
\right)^{\lambda_m/\lambda_0}$ &
$\displaystyle
\left(
1-\sum_{g\in\cG}\mu_g
\omega_{jg}
\right)\ln(s_j)
+\sum_{g\in\cG}\mu_g
\omega_{jg}
\ln\!\left(
\sum_{i\in\cJ}
\omega_{ig}s_i
\right)$ \\[0.5ex]
$\omega_{jm}\geq0$ $\forall j,m$,
$\sum_{m\in\mathcal M}\omega_{jm}>0$ $\forall j$ &
$\omega_{jg}\geq0\ \forall j,g$ \\
$0<\lambda_m\leq\lambda_0$ $\forall m$ &
$\sum_{g\in\cG}\mu_g\omega_{jg}<1\ \forall j$,
$\mu_g\geq0\ \forall g$ \\

\cmidrule(lr){1-2}
\end{tabular}}
\par\vspace{0.5ex}
\begin{minipage}{\textwidth}
\footnotesize\setstretch{1.0}\noindent \emph{Notes:} \emph{Notation:} $N$, $N_*$, and $N_g$ count groups; the $\lambda$ parameters belong to the GEV models and the $\mu$ parameters to our generators. Each row pairs a GEV model with an inverse market share model having the corresponding nesting structure, without implying identical demand systems. The paired systems coincide for logit and nested logit. In the two-level nested logit model, $g(j)$ denotes the nest containing product $j$. In the three-level nested logit model, $g_1(j)$ denotes the subnest containing product $j$, and $g_2(j)$ denotes the nest containing that subnest. In the product differentiation logit and IPDL models, $g$ indexes partitions, $\{\cG_{hg}\}_{h=1}^{N_g}$ is the collection of groups in partition $g$, and $h_g(j)$ indexes the unique group in partition $g$ that contains product $j$. In the cross-nested logit and partial membership models, $\omega_{jm}$ and $\omega_{jg}$ denote the corresponding nonnegative membership weights. Unlike cross-nested logit, partial membership allows
$\sum_{g\in\cG}\omega_{jg}=0$; then $\ln G_j(\vs)=\ln s_j$, giving
product $j$ the standard logit form.
\end{minipage}
\end{table}

\section{Identification and Estimation} \label{Sec:IdentandEstim}

Since the outside good belongs to no nest, $\omega_{0g}=0$ for every $g\in\cG$.
Then $\ln G_0(\vs_t)=\ln s_{0t}=\delta_{0t}-c_t$, and with the normalization
$\delta_{0t}=0$, $c_t=-\ln s_{0t}$.  Substituting into \eqref{eq:inversemarket} and
using a linear index \eqref{eq:Delta} yields the estimating equations
\begin{equation}\label{eq:lineariv2}
\ln\!\left(\frac{s_{jt}}{s_{0t}}\right)
=
\vx_{jt}^{\top}\vbeta
- \alpha p_{jt}
+ \sum_{g\in\cG}
\mu_g \omega_{jg}
\ln\!\left(
\frac{s_{jt}}
{\sum_{i\in\cJ} \omega_{ig} s_{it}}
\right)
+ \xi_{jt},
\qquad j=1,\dots,J.
\end{equation}

Under full membership ($\omega_{jg}=1\{j\in g\}$),
\eqref{eq:lineariv2} takes the same form as the estimating equations for the
logit, nested logit, and IPDL models, differing only in the specific
log-share terms implied by the chosen nesting structure.

Assume that characteristics $\vx_{jt}$ are exogenous, and prices and log-share terms are endogenous.

\paragraph{Fixed membership.}
With memberships fixed, equation \eqref{eq:lineariv2} is then a linear instrumental variable regression with $|\cG|+1$ endogenous variables: price and one log-share per group.
Let $\mathbf R_{jt}$ collect exogenous variables, price, and log-share terms, and $\mathbf Z_{jt}$ collect exogenous variables and excluded instruments.
The coefficients are identified if $\mathbb E[\mathbf Z_{jt}\xi_{jt}]=0$ and $\mathbb E[\mathbf Z_{jt}\mathbf R_{jt}^{\top}]$ has full column rank.
Exogenous variables instrument for themselves. Identification requires at least $|\cG|+1$ excluded instruments uncorrelated with $\xi_{jt}$ and generating sufficient independent variation in price and the inside-group log-share terms.

The literature has identified a range of possible instruments \cite[see, e.g.,][Section 4.2]{BH2021}.
Cost and markup shifters, including own and rival cost shifters and functions of exogenous rival characteristics (BLP-type instruments), that generate exogenous price variation can instrument prices. The log-share terms require instruments that shift a product's share relative to the nests it belongs to. As in \citet{fosgerau2024}, consider cereals segmented by market segment (kids, adults, all family) and by brand (e.g., General Mills and Kellogg's).
A price change for a General Mills kids cereal may shift demand toward cereals sharing its segment and brand, either, or neither.
The relative strength of these responses helps distinguish segment nesting effects from brand nesting effects.
Shocks affecting cereals in other segments and brands provide additional identifying variation. The same shifters therefore serve as instruments for the nesting parameters, including
membership-weighted sums of exogenous characteristics of rival products.

Given valid instruments, two-stage least squares or two-step GMM estimates $(\vbeta,\alpha,\{\mu_g:g\in\mathcal{G}\})$ in closed form, avoiding numerical inversion, starting values, and local optima. Computation is comparable to logit or nested logit for small $|\mathcal G|$, but increases with the number of nests.

\paragraph*{Estimated membership.}
Identifying variation must distinguish membership weights from nesting parameters, since both affect substitution.

{\emergencystretch=8em Rescaling $(\mu_g,\omega_{\cdot g})$ to $(\mu_g/\lambda,\lambda\omega_{\cdot g})$, for $\lambda>0$, leaves every coefficient $\mu_g\omega_{jg}$ unchanged and adds only the product-specific constant $\mu_g\omega_{jg}\ln\lambda$ to $\ln G_j$.
With product fixed effects, this constant is absorbed, so we normalize $\sum_j\omega_{jg}=1$ for each active inside nest.
Zero-intensity and singleton nests are omitted.
Nests with identical normalized memberships are combined, since only the sum of their parameters is identified.\par}

Equation \eqref{eq:lineariv2} is then a nonlinear instrumental variable regression. When differentiation under the expectation is valid, the parameters are locally identified if $\mathbb E[\mathbf Z_{jt}\xi_{jt}(\boldsymbol{\theta})]=0$ and $\mathbb E\!\left[\mathbf Z_{jt}\frac{\partial\xi_{jt}(\boldsymbol{\theta})}{\partial\boldsymbol{\theta}^{\prime}}\right]$ has full column rank at the truth, where $\boldsymbol{\theta}$ collects the free parameters after normalization. Price responses at different relative market shares help distinguish nesting parameters from memberships. In the cereal example, a strong response to a kids cereal's price may reflect strong segment nesting or substantial membership. Nesting parameters are segment-specific, whereas memberships vary across cereals. Excluded cost shocks across cereals, brands and segments change relative shares, revealing how responses vary with shares and helping separate common nesting effects from individual memberships.

\begin{samepage}
Given valid instruments, $\boldsymbol{\theta}$ can be estimated by nonlinear two-stage least squares or two-step GMM. Closed-form inverse shares again avoid numerical fixed points or contraction mappings. The generally nonconvex GMM objective makes optimization sensitive to starting values and local minima.\par
\end{samepage}

\section{Theoretical Foundations}\label{Section:Utility}\label{sec:perturbed-utility}

\begin{samepage}
For a convex function $\Psi:\Delta_J^\circ\to\R$ of shares $\vs$, called the perturbation, consider the utility maximization problem \citep{mcfadden2012theory}
\begin{equation}\label{eq:pum-full}
\max_{\vs\in\Delta_J^\circ}\{\vd^\top\vs-\Psi(\vs)\}.
\end{equation}
\end{samepage}

\setcounter{theorem}{0}
The following proposition establishes the equivalence between demand models defined by our generators and utility maximization \eqref{eq:pum-full}. The proof is given in Appendix~\ref{Appendix:ProofSec3}.

\begin{proposition}\label{prop:pum-generator}
Suppose $\Psi$, viewed as a function of $(s_1,\ldots,s_J)$, is twice continuously differentiable, has a positive definite
Hessian, and satisfies $\|\nabla\Psi(\vs)\|\to\infty$ as
$\vs\to\operatorname{bd}(\Delta_J)$. Then, for every $\vd\in\R^{J+1}$,
problem \eqref{eq:pum-full} has a unique maximizer satisfying \eqref{eq:inversemarket}
for a generator $\vG$. Conversely, for every generator $\vG$, the shares defined by \eqref{eq:inversemarket} uniquely solve \eqref{eq:pum-full} with
\[
\Psi(\vs)=\sum_{j\in\cJ}s_j\ln G_j(\vs),
\]
which satisfies the stated assumptions.
\end{proposition}
\setcounter{theorem}{0}

The forward direction extends \citet{fosgerauKristensenNielsen2021} to weaker boundary conditions, and the converse completes the equivalence. Our weaker boundary condition allows the utility penalty $\Psi$ to become infinite as some shares approach zero.

Problem \eqref{eq:pum-full} admits several interpretations. Consider first a representative consumer facing the choice set $\cJ$ of differentiated products and a homogeneous num\'eraire good priced at $1$, with income $y>\max_{j\in\cJ}p_j$, ensuring positive num\'eraire consumption. Each product $j\in\cJ$ has a quality $v_j\in\R$ and a price $p_j\in\R$, with $\delta_j=v_j-\alpha p_j$, where $\alpha>0$ is the marginal utility of income. We normalize $v_0=p_0=0$. The outside option belongs to the share vector, while the num\'eraire $z$ is separate consumption financed by the remaining budget.

The representative consumer chooses market shares $\vs\in\Delta_J^\circ$ and num\'eraire good consumption $z\geq0$ to maximize her quasi-linear direct utility $\alpha z+\sum_{j\in\cJ}v_js_j-\Psi(\vs)$ subject to the budget constraint $\sum_{j\in\cJ}p_js_j+z\leq y$. The term $-\Psi(\vs)$ captures taste for variety, traded off against quality. The binding budget constraint yields $\alpha y+\sum_j(v_j-\alpha p_j)s_j-\Psi(\vs)$. The additive income term $\alpha y$ leaves optimal shares unchanged, yielding \eqref{eq:pum-full}. For $\Psi(\vs)=\sum_js_j\ln G_j(\vs)$, the maximized utility is $\alpha y+CS(\vd)$.

The logit model uses $\Psi(\vs)=\sum_{j\in\cJ}s_j\ln s_j$ \citep{anderson1988representative}. The nested logit and IPDL models generalize this function \citep{verboven1996,fosgerau2024}. Our generator framework extends these constructions to $\Psi(\vs)=\sum_{j\in\cJ}s_j\ln G_j(\vs)$ on $\Delta_J^\circ$.

When $\vs$ denotes choice probabilities, \eqref{eq:pum-full} instead describes individual stochastic choice \citep{mcfadden2012theory,fudenberg2015stochastic,allen2023revealed}. Then $\Psi$ can represent preferences for randomization \citep{fudenberg2015stochastic}, information costs under rational inattention \citep{matejka2014rational,fosgerau2016discrete}, or costs of implementing a choice \citep{mattsson2002probabilistic}.

\citet{allen2019} show that, under their separability and independence assumptions, aggregate demand of heterogeneous consumers admits a representative utility representation of the form \eqref{eq:pum-full}. On finite covariate support, \eqref{eq:pum-full} admits a two-stage discrete choice representation with latent first-stage choices \citep{allen2026latent}.

Problem \eqref{eq:pum-full} also admits an additive random utility interpretation for a subclass of perturbations.
In an ARUM, consumers choose the product maximizing indirect utility $u_j=\delta_j+\varepsilon_j$, $j\in\cJ$. Suppose that the vector of random utility components $\boldsymbol{\varepsilon}=(\varepsilon_0,\ldots,\varepsilon_J)^\top$ has finite means and a strictly positive density on $\mathbb R^{J+1}$ and is independent of $\vd$. These assumptions exclude ties and random coefficients and ensure positive choice probabilities; $\varepsilon_j$ is the sole source of heterogeneity. Adding $\alpha y$ and using $\delta_j=v_j-\alpha p_j$ gives $\alpha y+v_j-\alpha p_j+\varepsilon_j$ without changing choices.

\citet{hofbauer2002global} show that continuously differentiable ARUM choice probabilities solve \eqref{eq:pum-full} with a perturbation satisfying the assumptions of Proposition~\ref{prop:pum-generator}, hence admit generator representation \eqref{eq:inversemarket}. They also show that \eqref{eq:pum-full} need not admit an ARUM representation. \citet{fosgerau2016discrete} derive the ARUM generator representation \eqref{eq:inversemarket} and its perturbation, and discuss extensions beyond ARUM. Proposition~\ref{prop:pum-generator} establishes equivalence between generators and perturbations under its assumptions, without requiring an ARUM. \citet{fosgerau2024} give two IPDL examples outside ARUM: one violates a restriction on mixed second-order price derivatives despite all products being substitutes, while the other exhibits complementarity in demand.

\section{Conclusion}\label{sec:ccl}

We develop a constructive approach to specifying demand models for market-level data through closed-form inverse market share functions. The nesting constructions accommodate rich substitution patterns through arbitrary nesting structures, and are consistent with utility maximization. With fixed memberships and a parsimonious set of nests, estimation requires only linear instrumental variable regression, while estimating membership weights requires nonlinear instrumental variables estimation, and a large set of nests creates empirical challenges \citep{hortaccsu2022estimating}. Future work includes extensions to dynamic and multi-category demand, and applications to markets with many products.

\appendix
\appendixpage

\section{Mathematical Appendix}\label{Appendix:Technical}

Vectors are columns:
$\vs=(s_0,\dots,s_J)^\top\in\R^{J+1}$. For $\vG:\R^{J+1}\to\R^{J+1}$, $\mathbf{J}^{\vs}_{\vG}(\tilde{\vs}) = \left[\partial G_i(\tilde{\vs})/\partial s_j\right]_{i,j}$ denotes its Jacobian at $\tilde{\vs}$. For $f:\R^{J+1}\rightarrow\R$, $\nabla_{\vs} f(\tilde{\vs}) = \left[\partial f(\tilde{\vs})/\partial s_i\right]_i$ and $\nabla^2_{\vs} f(\tilde{\vs}) = \left[\partial^2 f(\tilde{\vs})/\partial s_i\partial s_j\right]_{i,j}$ denote its gradient and Hessian. Scalar functions applied to vectors act coordinate-wise. $\vone_{J+1}=(1,\ldots,1)^\top\in\R^{J+1}$ denotes the vector of ones, and $\mathbf{I}_{J+1}\in\R^{(J+1)\times(J+1)}$ the identity matrix.

\subsection{Homogeneous Functions}\label{appendix:homo}

A function $G_j:\Rpp^{J+1}\to\Rpp$ is homogeneous of degree $\kappa>0$ if $G_j(\lambda\vs)=\lambda^\kappa G_j(\vs)$ for every $\lambda>0$ and $\vs\in\Rpp^{J+1}$. A function $\vG:\Rpp^{J+1}\to\Rpp^{J+1}$, with components $\vG(\vs)=(G_0(\vs),\ldots,G_J(\vs))^\top$, is homogeneous of degree $\kappa$ if each component $G_j$ is homogeneous of degree $\kappa$.

\begin{proposition}\label{prop:Euler_homo}
If $G:\Rpp^{J+1}\to\R$ is differentiable and homogeneous of degree $1$,
then
\[
G(\vs)
=
\sum_{i=0}^J \frac{\partial G(\vs)}{\partial s_i}\,s_i
\quad\text{for all }\vs\in\Rpp^{J+1}.
\]
\end{proposition}

\begin{proof}
See \citet[Proposition~4.3, p.~491]{carter2001foundations}. 
\end{proof}

\subsection{Properties of the Generator}\label{Appendix:Generator}

\begin{lemma}\label{Lemma:Euler}
For any generator $\vG$ and $\vs\in\Delta_J^\circ$,
\[
\sum_{j\in\cJ} s_j\,\frac{\partial\ln G_j(\vs)}{\partial s_k}
=
1\quad\forall\,k\in\cJ,
\]
or in matrix form,
$[\mathbf{J}_{\ln\vG}^{\vs}(\vs)]^\top\vs
=\vone_{J+1}$, equivalently $\mathbf{J}_{\ln\vG}^{\vs}(\vs)\vs=\vone_{J+1}$ by symmetry of $\mathbf{J}_{\ln\vG}^{\vs}(\vs)$, from Definition~\ref{Def:Generator}(iii).
\end{lemma}

\begin{proof}
For any $k\in\cJ$ and $\vs\in\Delta_J^\circ$,

\begin{equation*}
\sum_{j \in \cJ}s_{j}\frac{\partial\ln G_j\left(\vs\right)}{\partial s_{k}}  =  \sum_{j \in \cJ}s_{j}\frac{\partial\ln G_k\left(\vs\right)}{\partial s_{j}}
=   \dfrac{\sum_{j \in \cJ}s_{j}\frac{\partial G_k\left(\vs\right)}{\partial s_{j}}}{G_k\left(\vs\right)} = \dfrac{G_k\left(\vs\right)}{G_k\left(\vs\right)} = 1, 
\end{equation*}
where the first equality uses symmetry of $\mathbf{J}_{\ln \vG}^{\vs}(\vs)$ and the third uses Proposition~\ref{prop:Euler_homo}.

\end{proof}

\begin{samepage}
\begin{lemma}\label{prop:inversion}
A generator $\vG$ is invertible on \(\mathbb R_{++}^{J+1}\).
\end{lemma}

\begin{proof}
{\color{black}\hypersetup{allcolors=black}The proof extends the perturbation to a Legendre function on $\Rpp^{J+1}$, as in
\citet{fosgerauKristensenNielsen2021}. }Let $F(\mathbf q)=\sum_jq_j\ln G_j(\mathbf q)$ on $\Rpp^{J+1}$,
with restriction $\Psi$ to $\Delta_J^\circ$.
Lemma~\ref{Lemma:Euler}, extended by homogeneity to $\Rpp^{J+1}$, gives
$\partial F(\mathbf q)/\partial q_k=\ln G_k(\mathbf q)+1$.
Hence $\nabla F=\ln\vG+\vone_{J+1}$ and
$\nabla^2F=\mathbf J_{\ln\vG}^{\mathbf q}$, which is positive definite.
For $\vs\in\Delta_J^\circ$ and $\lambda>0$, homogeneity gives
$F(\lambda\vs)=\lambda\Psi(\vs)+\lambda\ln\lambda$,
since $\sum_js_j=1$.
Thus $F$ is twice continuously differentiable and strictly convex. Write $\mathbf q=\lambda\vs$, where $\lambda=\sum_jq_j$
and $\vs\in\Delta_J^\circ$. If $\mathbf q$ approaches a nonzero
boundary point, $\lambda$ has a positive limit and $\vs$ approaches
the simplex boundary. Since
$\nabla F(\mathbf q)=\ln\vG(\vs)+(1+\ln\lambda)\vone_{J+1}$,
Definition~\ref{Def:Generator}(ii) gives
$\|\nabla F(\mathbf q)\|\to\infty$.

If $\mathbf q\to\mathbf0$, then $\lambda\to0$.
Fix $\bar{\vs}\in\Delta_J^\circ$. Convexity gives
$F(\bar{\vs})\ge F(\vs)+\nabla F(\vs)^\top(\bar{\vs}-\vs)=\bar{\vs}^{\top}\ln\vG(\vs)$.
The right-hand side is a weighted average of the $\ln G_j(\vs)$,
so
\[
\min_j\ln G_j(\mathbf q)
=\ln\lambda+\min_j\ln G_j(\vs)
\le\ln\lambda+F(\bar{\vs})\to-\infty.
\]
Hence $\|\nabla F(\mathbf q)\|\to\infty$ also at the origin.
Thus the closed convex extension of $F$ is of Legendre type.

Define
$W(\mathbf y)=\sup_{\vs\in\Delta_J^\circ}
\{\mathbf y^\top\vs-\Psi(\vs)\}$.
Fix $\bar{\vs}\in\Delta_J^\circ$. Convexity gives
$\Psi(\vs)\geq\Psi(\bar{\vs})+\nabla F(\bar{\vs})^\top(\vs-\bar{\vs})$.
The right-hand side is affine and bounded below on the bounded simplex. Since $\mathbf y^\top\vs$ is also bounded there, $W(\mathbf y)$ is finite.

Writing $\mathbf q=\lambda\vs$ gives
\[
F^*(\mathbf y)=\sup_{\lambda>0}\{\lambda W(\mathbf y)-\lambda\ln\lambda\}
=e^{W(\mathbf y)-1}.
\]
The strictly concave expression is maximized at $\lambda=e^{W(\mathbf y)-1}$, where its derivative $W(\mathbf y)-\ln\lambda-1$ vanishes.

By \citet[Theorem~26.5]{rockafellar2015convex}, $\nabla F=\ln\vG+\vone_{J+1}$ is a bijection onto $\R^{J+1}$ with inverse $\nabla F^*$. Therefore $\vG$ is invertible.
\end{proof}
\end{samepage}

\begin{lemma}\label{lemma:H}
The inverse $\vH(e^{\vd})$ is continuously
differentiable and satisfies the following:
\begin{enumerate}
\item[(i)] $\vH$ is homogeneous of degree $1$;
\item[(ii)] For every \(j\in\mathcal J\), $H_j(e^{\vd})\rightarrow+\infty
\quad\text{as}\quad
\delta_j\to+\infty$
holding \(\delta_k\), \(k\neq j\), fixed;
\item[(iii)] $\mathbf{J}^{\vd}_{\vH \circ\exp}(\vd)$
is symmetric and positive definite for every $\boldsymbol{\delta}\in\R^{J+1}$.
\end{enumerate}
\end{lemma}

\begin{proof}
Use $F$ and $W$ from Lemma~\ref{prop:inversion}. Since $\nabla^2F$ is positive definite, the inverse function theorem makes $\nabla F^*$ continuously differentiable. Thus $F^*$ and $W=1+\ln F^*$ are twice continuously differentiable.

The definition of $W$ and $\sum_js_j=1$ give $W(\mathbf y+\tau\vone_{J+1})=W(\mathbf y)+\tau$. Differentiating in $\mathbf y$ and $\tau$ shows that $\nabla W$ is invariant to this translation and $\vone_{J+1}^\top\nabla W=1$. Since $F^*=e^{W-1}$ and $\nabla F^*=(\nabla F)^{-1}$,
\[
\vH(e^{\vd})=\nabla F^*(\vd+\vone_{J+1})
=e^{W(\vd)}\nabla W(\vd).
\]
Positivity of $\vH$ gives $\nabla W\in\Delta_J^\circ$. For (i), let $\lambda>0$. Then
\[
\begin{aligned}
\vH(\lambda e^{\vd})
&=e^{W(\vd+(\ln\lambda)\vone_{J+1})}
\nabla W(\vd+(\ln\lambda)\vone_{J+1})\\
&=\lambda e^{W(\vd)}\nabla W(\vd)
=\lambda\vH(e^{\vd}),
\end{aligned}
\]
which proves degree-one homogeneity.

For (iii), differentiating $\nabla F(\nabla F^*(\mathbf y))=\mathbf y$ gives
\[
\mathbf J_{\vH\circ\exp}^{\vd}(\vd)
=\nabla^2F^*(\vd+\vone_{J+1})
=\bigl[\nabla^2F(\vH(e^{\vd}))\bigr]^{-1},
\]
which is symmetric and positive definite.
For (ii), fix $\vd^0$ and let $a=\partial W(\vd^0)/\partial\delta_j>0$. With other indexes fixed, convexity gives $\partial W/\partial\delta_j\geq a$ for $\delta_j\geq\delta_j^0$. Convexity also gives $W(\vd)\geq W(\vd^0)+a(\delta_j-\delta_j^0)\to\infty$. Therefore $H_j(e^{\vd})=e^{W(\vd)}\partial W(\vd)/\partial\delta_j\geq a e^{W(\vd)}\to\infty$.
\end{proof}

\section{Proofs for Section~\ref{Section:Class}}
\label{Appendix:ProofSec2}

\subsection{Proof of the Statements in Subsection \ref{Sec:ClassModels}}\label{app:properties}

Use $F$ and $W$ from Lemma~\ref{prop:inversion}.
\begin{itemize}
\item[(i)] Write $T=\sum_kH_k(e^{\vd})$. Lemma~\ref{lemma:H} gives $T=e^{W(\vd)}$ and $\vs=\vH(e^{\vd})/T=\nabla W(\vd)$. Homogeneity gives $\vG(\vs)=e^{\vd}/T$, proving \eqref{eq:inversemarket} and \eqref{eq:demandsystem}. Translation of $W$ leaves its gradient unchanged, so $\boldsymbol\sigma(\vd+\tau\vone_{J+1})=\boldsymbol\sigma(\vd)$.
\item[(ii)] Put $\mathbf M=\mathbf J_{\ln\vG}^{\vs}(\vs)$. Since $\nabla^2F(T\vs)=\mathbf M/T$, differentiating $F^*(\vd+\vone_{J+1})=e^{W(\vd)}$ twice gives
$T\mathbf M^{-1}=T\bigl(\nabla^2W(\vd)+\vs\vs^\top\bigr)$.
Thus $\mathbf J_{\boldsymbol\sigma}^{\vd}=\nabla^2W=\mathbf M^{-1}-\vs\vs^\top$, which is symmetric positive semidefinite because $W$ is convex.
\item[(iii)] By (i), the function in \eqref{eq:CSfunc} equals $W=\max_{\vs\in\Delta_J^\circ}\{\vd^\top\vs-\Psi(\vs)\}$: the interior shares $\vs=\nabla W(\vd)$ satisfy the first-order conditions and uniquely maximize by strict concavity. It is increasing and convex, and $CS(\vd+\tau\vone_{J+1})=CS(\vd)+\tau$. By (i), $\ln G_0(\vs)=\delta_0-CS(\vd)$. Since the outside good belongs to no nest, $G_0(\vs)=s_0$.
With $\delta_0=0$, this gives $CS(\vd)=-\ln s_0$.
\end{itemize}

\subsection{Proof for Nesting Constructions}

\begin{proposition}\label{prop:direct-nesting}
Let $\cG\subseteq2^{\cJ}$ be a collection of nests with associated real parameters $(\mu_g)_{g\in\cG}$. Let $(\omega_{jg})_{j\in\cJ,\,g\in\cG}$ be nonnegative weights, with at least two positive weights in each nest.
For $\vG$ defined by \eqref{eq:ConstructionMember} on $\mathbb R_{++}^{J+1}$ to be a generator, a necessary condition is $b_j=1-\sum_g\mu_g\omega_{jg}>0$ for every $j\in\cJ$. A sufficient condition is
\[
\sum_{g\in\cG:\,\mu_g>0}\mu_g\omega_{jg}<1,
\qquad j\in\cJ.
\]
The two conditions coincide when all parameters are nonnegative. This includes construction \eqref{eq:generatorGN} when $\omega_{jg}=1\{j\in g\}$.
{\color{black}\hypersetup{allcolors=black}
For full membership with hierarchical nests, a necessary and sufficient condition for construction {\hypersetup{linkcolor=blue}\eqref{eq:generatorGN}} to define a generator is $\sum_{h\in\cG:\,g\subseteq h}\mu_h<1$ for every $g\in\cG$.
}
\end{proposition}

\begin{proof}
For sufficiency, omit groups with $\mu_g=0$. Write $\boldsymbol\omega_g=(\omega_{0g},\ldots,\omega_{Jg})^\top$, $S_g(\vs)=\boldsymbol\omega_g^\top\vs$, and $b_j^+=1-\sum_{g:\,\mu_g>0}\mu_g\omega_{jg}>0$. Exponentiating \eqref{eq:ConstructionMember} and separating positive and negative parameters gives
\begin{equation}\label{eq:signed-nesting-product}
G_j(\vs)=s_j^{b_j^+}
\prod_{g:\,\mu_g>0}S_g(\vs)^{\mu_g\omega_{jg}}
\prod_{g:\,\mu_g<0,\,\omega_{jg}>0}
\left(\frac{s_j}{S_g(\vs)}\right)^{|\mu_g|\omega_{jg}},
\end{equation}
which is positive and continuously differentiable on $\Rpp^{J+1}$. For (i), $S_g(\lambda\vs)=\lambda S_g(\vs)$ gives
\[
G_j(\lambda\vs)=\lambda^{b_j^++\sum_{g:\,\mu_g>0}\mu_g\omega_{jg}}G_j(\vs)=\lambda G_j(\vs),\qquad \lambda>0,
\]
since $b_j^++\sum_{g:\,\mu_g>0}\mu_g\omega_{jg}=1$.

For (iii), differentiating $\ln\vG(\vs)$ in \eqref{eq:ConstructionMember} with respect to $\vs$ gives
\[
\JlnG(\vs)=\operatorname{diag}\!\left(\frac{b_j}{s_j}\right)+\sum_g\mu_g\frac{\boldsymbol\omega_g\boldsymbol\omega_g^\top}{S_g(\vs)}.
\]
This matrix is symmetric. By Cauchy-Schwarz,
\[
\left(\sum_j\omega_{jg}x_j\right)^2\leq S_g(\vs)\sum_j\frac{\omega_{jg}x_j^2}{s_j}.
\]
Applying this bound to the negative-parameter terms and dropping the nonnegative terms gives, for $\mathbf x\ne\mathbf0$,
\[
\mathbf x^\top\JlnG(\vs)\mathbf x
\geq\sum_j\left(b_j-\sum_{g:\,\mu_g<0}|\mu_g|\omega_{jg}\right)\frac{x_j^2}{s_j}
=\sum_j b_j^+\frac{x_j^2}{s_j}>0.
\]

For (ii), let $\vs_n\in\Delta_J^\circ$ converge to $\vs^0\in\operatorname{bd}(\Delta_J)$, with $s_j^0=0$. Since $\sum_i s_{ni}=1$ and all weights are nonnegative, $S_g(\vs_n)\leq\max_i\omega_{ig}$ and $S_g(\vs_n)\geq\omega_{jg}s_{nj}$, so $s_{nj}/S_g(\vs_n)\leq1/\omega_{jg}$ whenever $\omega_{jg}>0$. Taking logarithms in \eqref{eq:signed-nesting-product} therefore gives
\[
\ln G_j(\vs_n)\leq b_j^+\ln s_{nj}+C_j\rightarrow-\infty\qquad\text{as }n\to\infty,
\]
where $C_j=\sum_{g:\,\mu_g>0}\mu_g\omega_{jg}\ln\max_i\omega_{ig}-\sum_{g:\,\mu_g<0,\,\omega_{jg}>0}|\mu_g|\omega_{jg}\ln\omega_{jg}$ is finite and independent of $n$. Since $b_j^+>0$ and $s_{nj}\to0$, $\|\ln\vG(\vs_n)\|\to\infty$.

For necessity, suppose $\vG$ is a generator. Fix $j$ and choose $\vs_n\in\Delta_J^\circ$ with $s_{nj}\to0$ and $s_{ni}\to s_i^0>0$ for every $i\ne j$. Each nest contains another product with positive weight, so every $S_g(\vs_n)$ has a positive limit. Thus, in
\[
[\JlnG(\vs_n)]_{jj}=\frac{b_j}{s_{nj}}+\sum_g\mu_g\frac{\omega_{jg}^2}{S_g(\vs_n)},
\]
the sum remains bounded. If $b_j<0$, $b_j/s_{nj}\to-\infty$, so $[\JlnG(\vs_n)]_{jj}\to-\infty$, contradicting condition (iii). If $b_j=0$, \eqref{eq:ConstructionMember} gives $\ln G_j(\vs_n)=\sum_g\mu_g\omega_{jg}\ln S_g(\vs_n)$, which remains bounded. For every $i\ne j$, $s_{ni}\to s_i^0>0$ and $S_g(\vs_n)\to S_g(\vs^0)>0$ give
\[
\ln G_i(\vs_n)=b_i\ln s_{ni}+\sum_g\mu_g\omega_{ig}\ln S_g(\vs_n)
\rightarrow b_i\ln s_i^0+\sum_g\mu_g\omega_{ig}\ln S_g(\vs^0)\in\R.
\]
Thus every component of $\ln\vG(\vs_n)$ remains bounded, contradicting condition (ii). Hence $b_j>0$.

{\color{black}\hypersetup{allcolors=black}
 {\color{black}\hypersetup{allcolors=black}For the necessary and sufficient condition, add $\cJ$ with parameter zero if absent. Let $\mathcal C(g)$ partition $g$ into its largest proper subnests and remaining singleton products.} Write $a_g=1-\sum_{h\in\cG:\,g\subseteq h}\mu_h$, $S_g=\sum_{j\in g}s_j$ and $Z_g=\sum_{j\in g}z_j$. For $\vs\in\Delta_J^\circ$ and $\mathbf z\in\mathbb R^{J+1}$,
\[
\mathbf z^\top\JlnG(\vs)\mathbf z=Z_{\cJ}^2+\sum_{g\in\cG}a_g Q_g,
\qquad Q_g=\sum_{c\in\mathcal C(g)}\frac{Z_c^2}{S_c}-\frac{Z_g^2}{S_g}\geq0.
\]
 {\color{black}\hypersetup{allcolors=black}Expanding the right-hand side gives coefficient $a_h-a_g=\mu_g$ on $Z_g^2/S_g$ for $g\ne\cJ$, where $h$ is the smallest nest strictly containing $g$; the coefficient for $\cJ$ is $1-a_{\cJ}=\mu_{\cJ}$. Singleton terms are $b_jz_j^2/s_j$, recovering the Jacobian.} Cauchy--Schwarz gives $Z_g^2=(\sum_{c\in\mathcal C(g)}Z_c)^2\leq S_g\sum_{c\in\mathcal C(g)}Z_c^2/S_c$, hence $Q_g\geq0$.

 {\color{black}\hypersetup{allcolors=black}Summing $Q_g$ cancels all intermediate nest terms, leaving $\sum_g Q_g=\sum_j z_j^2/s_j-Z_{\cJ}^2$.} If all $a_g>0$, then $m=\min(1,\min_g a_g)>0$ gives $\mathbf z^\top\JlnG(\vs)\mathbf z\geq m\sum_j z_j^2/s_j$, proving (iii).

Let $c_g(j)\in\mathcal C(g)$ contain $j$. Equation {\hypersetup{linkcolor=blue}\eqref{eq:generatorGN}} also gives
\[
\ln G_j(\vs)=\sum_{g\in\cG:\,j\in g}a_g\ln\frac{S_{c_g(j)}}{S_g}\leq m\ln s_j.
\]
 {\color{black}\hypersetup{allcolors=black}The ratios are at most one and multiply to $s_j$, so $s_j\to0$ implies $\ln G_j\to-\infty$, proving (ii). Positivity, continuous differentiability and homogeneity follow from {\hypersetup{linkcolor=blue}\eqref{eq:generatorGN}}.}

 {\color{black}\hypersetup{allcolors=black}For necessity, fix $g$ and choose distinct $a,b\in\mathcal C(g)$. Set $z_j=s_j/S_a$ on $a$, $z_j=-s_j/S_b$ on $b$, and zero elsewhere. Then $Z_a=1$, $Z_b=-1$ and $Z_g=0$. Proportionality within each group and cancellation in larger nests give $Q_h=0$ for $h\ne g$. Thus the quadratic form equals $a_g(1/S_a+1/S_b)$, and positive definiteness requires $a_g>0$.}
}

\end{proof}

\subsection{Proof of Subsection~\ref{sec:substitution-patterns}}\label{app:substitution-patterns}\label{app:isolated-outside}

Let $\mathbf M=\JlnG(\vs)$, $\mathbf D=\operatorname{diag}(\vs)$, and $\mathbf A=\mathbf M^{-1}-\vs\vs^\top$.

\paragraph{Proof of Proposition~\ref{prop:substitution-structure}.}
(i) The restriction $\sum_{g:\,j\in g}\mu_g<1$ gives $b_j>0$. Start with $\operatorname{diag}(b_j/s_j)$ and add the nest terms $(\mu_g/S_g)\boldsymbol\omega_g\boldsymbol\omega_g^\top$ in increasing order of nest size. The matrix $\operatorname{diag}(b_j/s_j)+\sum_{g\in\mathcal H}(\mu_g/S_g)\boldsymbol\omega_g\boldsymbol\omega_g^\top$, where $\mathcal H$ contains the nests already added, is block diagonal. The blocks correspond to maximal nests already added and remaining singleton products. The nesting hierarchy makes each new nest a union of these blocks.

We show by induction that each inverse block has nonpositive off-diagonal entries and positive row sums. This holds initially because the inverse is $\operatorname{diag}(s_j/b_j)$. For the next nest $g$, let $\mathbf B$ be its submatrix and $c=\mu_g/S_g\geq0$. Since $g$ contains whole blocks, induction gives $\mathbf u=\mathbf B^{-1}\mathbf1>0$. Full membership makes the nest term $c\mathbf1\mathbf1^\top$ on $g$ and zero elsewhere, replacing $\mathbf B$ by $\mathbf B+c\mathbf1\mathbf1^\top$. The Sherman-Morrison formula gives
\begin{equation}\label{eq:nesting-sherman-morrison}
(\mathbf B+c\mathbf1\mathbf1^\top)^{-1}
=\mathbf B^{-1}-\frac{c\mathbf u\mathbf u^\top}{1+c\mathbf1^\top\mathbf u}.
\end{equation}
The subtracted matrix has nonnegative entries, so the off-diagonal entries remain nonpositive. Multiplying by $\mathbf1$ gives positive row sums, $\mathbf u/(1+c\mathbf1^\top\mathbf u)>0$. After adding all nests, the matrix is $\mathbf M$, so induction gives $[\mathbf M^{-1}]_{jk}\leq0$ and hence $A_{jk}=[\mathbf M^{-1}]_{jk}-s_js_k\leq-s_js_k$ for $j\ne k$.

(ii) Let $g^+,h^+$ denote the supports of the two specified nests. Since neither contains the other, choose $j\in g^+\setminus h^+$ and $k\in h^+\setminus g^+$. Fix $\vs\in\Delta_J^\circ$ and choose $\mu_g,\mu_h>0$ small enough that every $b_i>0$, setting all other parameters to zero. Let $d_i=s_i/b_i$ and $\mathbf B=\operatorname{diag}(1/d_i)+(\mu_g/S_g)\boldsymbol\omega_g\boldsymbol\omega_g^\top$. Our choice of $j,k$ gives $j\notin h^+$ and $k\notin g^+$, so $\omega_{jh}=\omega_{kg}=0$. Applying \eqref{eq:nesting-sherman-morrison} with $\boldsymbol\omega_g$ in place of $\mathbf1$ gives
\[
\begin{aligned}
{}[\mathbf B^{-1}]_{jk}
&=-\frac{\mu_g d_jd_k\omega_{jg}\omega_{kg}}{S_g+\mu_g\sum_i d_i\omega_{ig}^2}=0,\\
(\mathbf B^{-1}\boldsymbol\omega_h)_j
&=-\frac{\mu_g d_j\omega_{jg}\sum_i d_i\omega_{ig}\omega_{ih}}{S_g+\mu_g\sum_i d_i\omega_{ig}^2}<0,
\qquad (\mathbf B^{-1}\boldsymbol\omega_h)_k=d_k\omega_{kh}>0.
\end{aligned}
\]
Overlap and nonnegative memberships ensure $\sum_i d_i\omega_{ig}\omega_{ih}>0$. Applying \eqref{eq:nesting-sherman-morrison} again, with $\boldsymbol\omega_h$ in place of $\mathbf1$, gives
\[
[\mathbf M^{-1}]_{jk}
=-\frac{\mu_h(\mathbf B^{-1}\boldsymbol\omega_h)_j(\mathbf B^{-1}\boldsymbol\omega_h)_k}
{S_h+\mu_h\boldsymbol\omega_h^\top\mathbf B^{-1}\boldsymbol\omega_h}>0.
\]
Let $U=g^+\cup h^+$. Since $0\notin U$, we can set $s_i(t)=ts_i$ on $U$ and redistribute the remaining share positively outside $U$. With only $g,h$ active, the isolated $U$ block of $\mathbf M$ scales by $1/t$, since $s_i,S_g,S_h$ scale by $t$. Thus
\[
A_{jk}(\vs(t))=t[\mathbf M(\vs)^{-1}]_{jk}-t^2s_js_k>0
\]
for sufficiently small $t>0$. Giving the other nests sufficiently small positive parameters preserves admissibility and this strict inequality by continuity.
\qed\par

\paragraph{Proof of Proposition~\ref{prop:substitution-expansion}.}
Let $L_{jk}=-w_{jk}$ for $j\ne k$ and $L_{jj}=w_j$. The definitions of $w_{jk}$ and $w_j$, using $\sum_k\omega_{kg}s_k=S_g$, give
\[
M_{jk}=\frac{w_{jk}}{s_js_k}\quad(j\ne k),\qquad
M_{jj}=\frac{b_j}{s_j}+\sum_g\frac{\mu_g\omega_{jg}^2}{S_g}
=\frac1{s_j}-\frac{w_j}{s_j^2}.
\]
Thus $\mathbf D\mathbf M\mathbf D=\mathbf D-\mathbf L$, so  $\mathbf D\mathbf M=\mathbf I-\mathbf L\mathbf D^{-1}$. Multiplying on the right by $\mathbf M^{-1}$ and rearranging yields $\mathbf M^{-1}=\mathbf D+\mathbf L\mathbf D^{-1}\mathbf M^{-1}$. Substituting this expression twice gives
\[
\mathbf M^{-1}
=\mathbf D+\mathbf L+\mathbf L\mathbf D^{-1}\mathbf L
+(\mathbf L\mathbf D^{-1})^3\mathbf M^{-1}.
\]
At fixed shares and memberships, linearity gives $\|\mathbf L\|\leq C\|\boldsymbol\mu\|$ for some $C$, and $\mathbf D^{-1}$ is fixed. Also, $\mathbf M^{-1}\to\mathbf D$ as $\boldsymbol\mu\to\mathbf0$, so its norm is bounded. Hence, using the operator norm,
\[
\bigl\|(\mathbf L\mathbf D^{-1})^3\mathbf M^{-1}\bigr\|
\leq\|\mathbf L\|^3\|\mathbf D^{-1}\|^3\|\mathbf M^{-1}\|
=O(\|\boldsymbol\mu\|^3).
\] Taking off-diagonal entries and subtracting $s_js_k$ proves the result.
\qed\par

\paragraph{Proof of the direct and indirect effect formulas.}
 Let $\mathbf e$ have entries $e_j=1$, $e_k=-1$, and zero elsewhere. Holding other shares fixed gives $ds_k=-ds_j$, hence $d\vs=\mathbf e\,ds_j$. Differencing \eqref{eq:inversemarket} for $j$ and $k$ gives $d\delta_j-d\delta_k=(\mathbf e^\top\mathbf M\mathbf e)\,ds_j$. Changing $p_k$ while holding $p_j$ fixed gives $d\delta_j-d\delta_k=\alpha\,dp_k$. Hence $S^d_{jk}=\alpha/(\mathbf e^\top\mathbf M\mathbf e)>0$, giving the stated formula after substituting $\mathbf M$. The indirect effect is the total effect $-\alpha A_{jk}$ minus the direct effect, as in \citet{ogaki1990}. By Proposition~\ref{prop:substitution-expansion},
\[
S^i_{jk}=\alpha\left[s_js_k+w_{jk}-\sum_{i\ne j,k}\frac{w_{ji}w_{ik}}{s_i}+w_{jk}\left(\frac{w_j}{s_j}+\frac{w_k}{s_k}\right)\right]-S^d_{jk}+O(\|\boldsymbol\mu\|^3),
\]
with $S^d_{jk}$ exact. At logit, $S^i_{jk}=-\alpha s_js_k(1-s_j-s_k)/(s_j+s_k)<0$ although every $w_{jk}$ vanishes. Thus Ogaki's indirect effect exists even without nesting.\qed\par

\paragraph{Proof of Proposition~\ref{prop:combined-flexibility}.}

Let $\mathbf B=\mathbf A^*+\vs\vs^\top$. If $\mathbf x^\top\mathbf B\mathbf x=\mathbf x^\top\mathbf A^*\mathbf x+(\vs^\top\mathbf x)^2=0$, then $\mathbf A^*\mathbf x=\mathbf0$, so $\mathbf x=c\vone_{J+1}$ and $c=\vs^\top\mathbf x=0$. Hence $\mathbf B$ is positive definite, and $\mathbf A=\mathbf A^*$ if and only if $\mathbf M=\mathbf B^{-1}$. Since $\mathbf B\vone_{J+1}=\vs$ and $B_{0j}=A^*_{0j}+s_0s_j=0$ for $j\geq1$, we have $\mathbf B^{-1}\vs=\vone_{J+1}$ and $[\mathbf B^{-1}]_{0j}=0$. With pairwise nests, differentiating \eqref{eq:ConstructionMember} gives $[\JlnG(\vs)]_{0j}=0$ and $[\JlnG(\vs)]_{jk}=\mu_{jk}/(s_j+s_k)$ for $1\leq j<k\leq J$, and $b_j+\sum_{k\neq j}\mu_{jk}=1$ gives $\JlnG(\vs)\vs=\vone_{J+1}$. A symmetric matrix $\mathbf X$ with $\mathbf X\vs=\vone_{J+1}$ is determined by its off-diagonal entries, since $X_{jj}=(1-\sum_{k\neq j}X_{jk}s_k)/s_j$. Hence $\mathbf M=\mathbf B^{-1}$ if and only if $\mu_{jk}=(s_j+s_k)[\mathbf B^{-1}]_{jk}$ for $1\leq j<k\leq J$. Proposition~\ref{prop:direct-nesting} gives the stated conditions: each pair has two positive weights, and the isolated outside good has $b_0=1$.\qed\par

\section{Proof of Proposition~\ref{prop:pum-generator}}\label{Appendix:ProofSec3}

\par\noindent
For $\mathbf q\in\Rpp^{J+1}$, put $\lambda=\sum_jq_j$, $\vs=\mathbf q/\lambda$, and define
\[
F(\mathbf q)=\lambda\Psi(\vs)+\lambda\ln\lambda,
\qquad \vG(\mathbf q)=\exp\bigl(\nabla F(\mathbf q)-\vone_{J+1}\bigr).
\]
Since $\Psi$ is twice continuously differentiable, so is $F$; hence $\vG$ is positive and continuously differentiable. Differentiating $F(t\mathbf q)=tF(\mathbf q)+t\ln t\sum_jq_j$ with respect to $\mathbf q$ and dividing by $t$ gives $\nabla F(t\mathbf q)=\nabla F(\mathbf q)+(\ln t)\vone_{J+1}$. Exponentiating gives $\vG(t\mathbf q)=t\vG(\mathbf q)$, proving condition (i).

For the Hessian, $\partial\lambda/\partial q_k=1$ and $\partial s_i/\partial q_k=(1\{i=k\}-s_i)/\lambda$ give
\[
\frac{\partial F}{\partial q_0}
=\Psi-\sum_{i=1}^Js_i\frac{\partial\Psi}{\partial s_i}+\ln\lambda+1,
\qquad
\frac{\partial F}{\partial q_j}
=\frac{\partial F}{\partial q_0}+\frac{\partial\Psi}{\partial s_j},\quad j\geq1.
\]
Here $F$ and $\Psi$ are evaluated at $\mathbf q$ and $\vs$, respectively. Differentiating again gives
\[
\frac{\partial}{\partial s_a}\left(\Psi-\sum_{i=1}^Js_i\frac{\partial\Psi}{\partial s_i}\right)
=\frac{\partial\Psi}{\partial s_a}-\frac{\partial\Psi}{\partial s_a}
-\sum_{i=1}^Js_i\frac{\partial^2\Psi}{\partial s_i\partial s_a}
=-\sum_{i=1}^Js_i\frac{\partial^2\Psi}{\partial s_i\partial s_a}.
\]
At $\mathbf q=\vs$, where $\lambda=1$, the chain rule gives
\[
\frac{\partial^2F}{\partial q_k\partial q_\ell}
=1+\sum_{i,a=1}^J(1\{i=k\}-s_i)
\frac{\partial^2\Psi}{\partial s_i\partial s_a}(1\{a=\ell\}-s_a).
\]
Let $I=\{1,\ldots,J\}$ and $\tau=\sum_jz_j$. Multiplying by $z_kz_\ell$ and summing over $k$ and $\ell$ and using $\sum_kz_k(1\{i=k\}-s_i)=z_i-\tau s_i$ yields
\[
\mathbf z^\top\nabla^2F(\vs)\mathbf z
=\tau^2+(\mathbf z_I-\tau\vs_I)^\top\nabla^2\Psi(\vs)(\mathbf z_I-\tau\vs_I).
\]
This can vanish only if $\tau=0$ and $\mathbf z_I=\mathbf0$, which also give $z_0=\tau-\sum_{j=1}^Jz_j=0$. Thus $\mathbf J_{\ln\vG}=\nabla^2F$ is symmetric positive definite, proving condition (iii). On the simplex, $\lambda=1$, so $F(\vs)=\Psi(\vs)$. Since $s_0=1-\sum_{i=1}^Js_i$, the chain rule gives
\[
\frac{\partial\Psi}{\partial s_j}
=\frac{\partial F}{\partial q_j}-\frac{\partial F}{\partial q_0}
=\ln G_j-\ln G_0,\qquad j=1,\ldots,J,
\]
with derivatives evaluated at $\mathbf q=\vs$. For each $j=1,\ldots,J$, the triangle inequality gives
\[
\left|\frac{\partial\Psi(\vs)}{\partial s_j}\right|
=|\ln G_j(\vs)-\ln G_0(\vs)|
\leq|\ln G_j(\vs)|+|\ln G_0(\vs)|
\leq2\|\ln\vG(\vs)\|.
\]
Summing the squared component bounds gives $\|\nabla\Psi(\vs)\|^2=\sum_{j=1}^J|\partial\Psi(\vs)/\partial s_j|^2\leq4J\|\ln\vG(\vs)\|^2$. Taking square roots gives $\|\nabla\Psi(\vs)\|\leq2\sqrt J\|\ln\vG(\vs)\|$. Gradient divergence therefore forces $\|\ln\vG(\vs)\|\to\infty$ at the boundary, proving condition (ii). Hence $\vG$ is a generator. For any $\vd$, its interior shares in \eqref{eq:demandsystem} satisfy \eqref{eq:inversemarket}, hence $\delta_j-\delta_0=\partial\Psi/\partial s_j$. These are the first-order conditions of \eqref{eq:pum-full}; strict concavity makes these shares the unique maximizer.

Conversely, restrict the Legendre function $F=\sum_jq_j\ln G_j$ of Lemma~\ref{prop:inversion} to the simplex. Its restriction $\Psi$ is twice continuously differentiable. Let $\mathbf B$ be the Jacobian of $(s_1,\ldots,s_J)\mapsto(1-\sum_{j=1}^Js_j,s_1,\ldots,s_J)$. The chain rule gives
\[
\nabla\Psi=(\ln G_j-\ln G_0)_{j=1}^J,\qquad
\nabla^2\Psi=\mathbf B^\top\nabla^2F\,\mathbf B,
\quad \mathbf B=\begin{pmatrix}-\vone_J^\top\\\mathbf I_J\end{pmatrix}.
\]
The Hessian is positive definite because $\nabla^2F$ is positive definite and $\mathbf B$ has full column rank. Continuity and strict positivity of $\boldsymbol\sigma$ on compact sets imply that bounded indexes keep every share bounded away from zero. Hence $\vs=\boldsymbol\sigma(0,\nabla\Psi(\vs))$ implies $\|\nabla\Psi(\vs)\|\to\infty$ as $\vs$ approaches $\operatorname{bd}(\Delta_J)$. The same first-order conditions prove the converse.\qed\par

\end{document}